\documentclass[submission,copyright,creativecommons]{eptcs}
\providecommand{\event}{PROLE 2015} 

\usepackage[activeacute,spanish,english]{babel}
\usepackage[all]{xy}
\usepackage[utf8]{inputenc}
\usepackage{graphicx} 
\usepackage{latexsym}
\usepackage{amsfonts}
\usepackage{amsmath} 
\usepackage{amsthm}
\usepackage{amssymb} 
\usepackage{tikz}
\usetikzlibrary{shapes.geometric,arrows}
\usetikzlibrary{arrows,automata,matrix,scopes}
\usetikzlibrary{positioning}
\usepackage{bbding}

\usepackage{float} 
\usepackage[section]{placeins}

\usepackage{eurosym}  
\usepackage{dsfont}       
\usepackage{lipsum}
\usepackage{paralist} 

\newtheorem{definition}{Definition}
\newtheorem{theorem}{Theorem}
\newtheorem{proposition}{Proposition}
\newtheorem{lemma}{Lemma}
\newtheorem{example}{Example}
\newtheorem{corollary}{Corollary}
\newtheorem{remark}{Remark}

\begin{document}
\title{Proving Continuity of Coinductive Global Bisimulation Distances: A Never
Ending Story\thanks{This work was Partially supported by $^{1,2}$ The Spanish project STRONGSOFT (TIN2012-39391-C04-04)  and UCM-Santander grant GR3/14. \newline$^{2}$ The Spanish project N-GREENS Software-CM (S2013/ICE-2731). \newline$^{3}$ The Icelandic project \emph{Processes and Modal Logics} (project nr.~100048021) and \emph{Decidability and Expressiveness for Interval Temporal Logics} (project nr.~130802-051) of the Icelandic Research Fund. \newline$^{1,2,3}$ The EEA Grants from the project \emph{Formal Methods for the Development and Evaluation of Sustainable Systems} (001-ABEL-CM-2013). }}
\author{David Romero-Hern\'andez $^{1}$ \quad\qquad David de Frutos-Escrig $^{2}$ \quad\qquad Dario Della Monica $^{3}$
\institute{$^{1,2}$ Facultad CC. Matem{\'a}ticas, Universidad Complutense de Madrid, Madrid, Spain.\\ Departamento de Sistemas Inform\'aticos y Computaci\'on} \institute{$^{3}$ ICE-TCS, School of Computer Science, Reykjavik University, Reykjavik, Iceland
} \email{dromeroh@pdi.ucm.es \quad\qquad defrutos@sip.ucm.es \quad\qquad dariodm@ru.is}}
\def\titlerunning{Proving Continuity of Coinductive Global Bisimulation Distances: A Never-ending Story}
\def\authorrunning{D. Romero-Hern\'andez, D. de Frutos-Escrig \& D. Della Monica}

\maketitle 

\pagestyle{empty}

\begin{abstract}
We have developed a notion of global bisimulation distance between processes which goes somehow beyond the notions of bisimulation distance already existing in the literature, mainly based on bisimulation games. Our proposal is based on the cost of transformations: how much we need to modify one of the compared processes to obtain the other. Our original definition only covered finite processes, but a coinductive approach allows us to extend it to cover infinite but finitary trees. After having shown many interesting properties of our distance, it was our intention to prove continuity with respect to projections, but unfortunately the issue remains open. Nonetheless, we have obtained several partial results that are presented in this paper.
\end{abstract}

\section{Introduction}
The notion of bisimulation has been extensively used to characterize the equivalence between processes \cite{Mil89, Par81, San11}. Bisimulations are coinductive proofs of that equivalence, which is called the bisimilarity relation. Certainly, bisimilarity is a quite natural relation, as suggested by the existence of several different formulation of the notion of bisimulation, e.g., in terms of bisimulation games~\cite{Sti98}. 

Up to bisimilarity, the semantics of processes is characterized by unordered trees without repeated (equivalent) branches, which are thus considered the canonical semantic model.
Therefore, two processes are equivalent if, and only if, they have the same semantic tree. But when two processes are not equivalent we have no way of expressing ``how different'' they are. Recently, several notions of bisimulation distance have been proposed, based on variants of the bisimulation game \cite{chr10_quantitative, amrs07, dlt08, ftl11}: while the original game imposes to the defender the obligation of replicating exactly any move by the attacker, in these variants the defender has the possibility of ``cheating'', by replying an attacker's move by choosing similar, but not equal, actions. However, when doing that, the defender have to pay a price according to the distance between the two involved actions.

This is a very suggestive path to follow when defining a bisimulation distance, and it comes with several efficient ways to compute it. Despite such desirable properties, we believe that alternative approaches are possible and worth being studied.
Previous works~\cite{rf12a, rf12b, rf14} by the first two authors of this paper contain several examples which mainly show that the ``classical'' distances based on variants of the bisimulation game are \emph{local}, in the sense that they only capture the difference between a single pair of executions of the two processes, thus failing in characterizing the distance between the processes in their entirety.

In the quest for a \emph{global} notion of distance, considering all executions at the same time, we have proposed a novel approach~\cite{rf12a, rf12b, rf14}: since trees express the bisimulation semantics, we looked for a natural distance between trees that defines what we called \emph{global bisimulation distance}.
For this purpose, we defined \emph{atomic transformations} between processes; then, a sequence of transformations provides an upper bound for the distance between the two processes at the two ends of the sequence.

This approach is suitable for comparing finite processes, but it is clearly inadequate when comparing infinite finitary processes with infinitely many differences. We are interested in obtaining sound bounds also in the latter case, whenever the series collecting all those differences converges. Instead of looking for a complex scenario based on the use of limits, we introduced in \cite{rf14} a coinductive framework which allows us to obtain bounds for those distances in a very simple way (coinduction is sometimes presented as an ``inductionless induction'' mechanism).

As we said before, classical bisimulation distances are easy to calculate, even (or we should better say, especially) in the quantitative cases (e.g., probabilistic \cite{bw05}, timed \cite{Bre05}), where calculus provides the machinery to obtain the corresponding fixed points. It is true that the computation of our (bounds for the) distance requires specific techniques in each case, but our coinductive approach benefits from the power of coinduction to accomplish this task.

In order to give a broader support to our approach, it was our intention to prove the continuity of our distance: we expected that whenever all the pairs of projections of two (possibly infinite) processes are at some fixed distance, the (full) processes themselves will be at that distance. Unfortunately, this paper tells an unfinished story: we were confident about obtaining a perhaps bit involved, but somehow ``standard'' proof, but our creature has revealed itself as an irresistible beast. Our, more and more, sophisticated attempts to domesticate it have crashed over and over, revealing new faces of the beast, whose actual existence remains uncertain. We all know how difficult is to disprove the existence of Nessy, Bigfoot, or E.T. Certainly, finding indisputable evidence of their existence would put an end to the mystery, but still we believe that this is impossible, as they simply do not exist \dots or do they?

Let us go with our story. It is not easy to tell unfinished stories, but we think that we have to do it, because we really enjoyed many exciting adventures in pursuing our quest and because we feel we might have paved the way for somebody else's success.

\section{Classic and global bisimulation distances} \label{ClassicalAndGlobal}
Our starting point will be the operational definition of processes as \emph{Labelled Transition Systems (lts)} over an alphabet $\mathds A$, which are tuples $(N,\mathit{succ})$, where $N$ is a set of states and $\mathit{succ}:N\rightarrow \mathcal{P}(\mathds{A}\times N)$. When we want to distinguish an initial state $n_0 \in N$, we write $(N,\mathit{succ},n_0)$. Sometimes we omit the component $\mathit{succ}$, simply considering $(N,n_0)$. Finite computations or \emph{paths}, are sequences $n_0a_1n_1\dots a_kn_k$ with $(a_{i+1},n_{i+1})\in \mathit{succ}(n_i)$ $\forall i \in\{0\dots k-1\}$. We denote the set of paths of an $lts$ $(N,n_0)$ by $Path(N,n_0)$.

\begin{definition}
We say that an lts $(N,n_0)$ is (or defines) a \emph{tree} $t$ if for all $n\in N$ there is a single path $n_0a_1n_1\dots a_jn_j$ with $n_j=n$. 
Then, we say that each node $n_k$ is at \emph{level $k$} in $t$, and define $Level_k(t)=\{n\in N \mid n$ is at level $k$ in $t\}$. We define the \emph{depth} of $t$ as $depth(t)=sup\{l\in\mathds{N}\mid Level_l(t)\neq\emptyset\}\in \mathds{N}\cup\{\infty\}$. We denote by $Trees(\mathds{A})$ the class of trees on the alphabet $\mathds{A}$, and by $FTrees(\mathds{A})$, the subclass of finite trees. 
\end{definition}

Every node $n$ of a tree $t=(N,n_0)$ induces a subtree $t_n=(N_n,n)$, where $N_n$ is the set of nodes $n^\prime_k\in N$ such that there exists a path $n^\prime_0a_1n^\prime_1\dots a_kn^\prime_k$ with $n^\prime_0=n$. We decompose any tree $t$ into the formal sum $\sum_{n_{1j}\in Level_1(t)} a_jt_{n_{1j}}$ (we denote the empty sum by {\bf 0}). Since our trees are unordered, by definition, this formal sum is also unordered. 

For each tree $t\in Trees(\mathds{A})$, we define its \emph{first-level width}, that we represent by $||t||$, as $||t||=|Level_1(t)|$. We also define the \emph{first $k$-levels width of $t$}, denoted by $||t||_k$, as $||t||_k=max\{||t_n||\mid n\in \bigcup_{l < k} Level_l(t)\}$. \emph{Finitary trees} are those such that $||t||_k < \infty$, $\forall k \in \mathds{N}$. We denote by $FyTrees(\mathds{A})$ the collection of  \emph{finitary trees} in $Trees(\mathds{A})$.

\begin{definition}\sloppy
Given an lts with initial state $(N,\mathit{succ},n_0)$, we define its unfolding, denoted by $\it{unfold}(N)$, as the tree $(\overline{N}, \overline{\mathit{succ}}, \overline{n_0})$, where $\overline{N}=Path(N,n_0)$, $\overline{\mathit{succ}}(n_0a_1\dots n_k)=\{(a,n_0a_1\dots n_k a n^\prime)\mid (a,n^\prime)\in \mathit{succ}(n_k)\}$, and $\overline{n_0}=n_0$.  An lts is finitely branching when its unfolding is a finitary tree.
\end{definition}

%
\begin{definition}
Given a tree $t=(N,\mathit{succ},n_0)$ and $k\in\mathds{N}$, we define its $k$-th cut or projection, denoted by $\pi_k(t)$, as the restriction of $t$ to the nodes in $\bigcup_{l\leq k} Level_l(t)$:
\begin{center}
$\pi_k(t)=(\pi_k(N),\mathit{succ}_k,n_0)$, where $\pi_k(N)=\bigcup_{l\leq k} Level_l(t)$, $\mathit{succ}_k(n)=\mathit{succ}(n)$ if $n\in \bigcup_{l<k} Level_l(t)$, and $\mathit{succ}_k(n)=\emptyset$ if $n\in Level_k(t)$.
\end{center}
\end{definition}

In this paper, we focus on finitely branching lts, and thus on finitary trees.
Each finitary tree is univocally defined by the sequence of its projections: $\forall t,t^\prime\in FyTree(\mathds{A}) ((\forall k\in \mathds{N} \pi_k(t) \! = \! \pi_k(t^\prime))  \! \Rightarrow \! t \! = \! t^\prime)$.

We consider domains of actions $(\mathds{A},{\bf d})$, where ${\bf d}:\mathds{A}\times \mathds{A} \rightarrow \mathds{R}^+\cup\{\infty\}$ is a distance between actions, with ${\bf d}(a,b)={\bf d}(b,a)$, ${\bf d}(a,b)=0$ $\Leftrightarrow$ $a=b$, and ${\bf d}(a,c)+{\bf d}(c,b)\geq {\bf d}(a,b)$, $\forall a,b,c \in \mathds{A}$, where $+$ is extended to $\mathds{R}^+\cup\{\infty\}$ as usual.  We assume that ${\bf d}(a,b)=\infty$ when the value ${\bf d}(a,b)$ is not specified.

When comparing pairs of  processes, it is natural \cite{chr10_distance,ahm03,ftl11} to introduce a \emph{discount factor} $\alpha \in (0,1]$. Then, the differences in the $k$-th level of the compared trees are weighted by $\alpha^k$, following the idea that differences in the far future are less important than those in the near. As a consequence, it is possible to obtain finite distances when comparing two  processes with infinitely many differences.

In \cite{rf12a}, we have presented, inter alia, an operational definition that allows us to obtain bounds for our global bisimulation distance between finite trees. These bounds are given by the cost of any transformation that turns one of the trees into the other. The following definition states which are the valid steps of those transformations and their costs. Roughly, any application of idempotency of $+$ has no cost, while the change of an action $a$ at level $k$ into another $b$, costs $\alpha^k{\bf d}(b,a)$. 
Intuitively, in what follows we write $t\rightsquigarrow^1_{\alpha,d} t'$
meaning that there is a distance step (aka 1-step transformation) between $t$
and $t'$ (with discount factor $\alpha$) whose cost is at most $d$. We use the
superscript 1 to distinguish between the 1-step relation
$\rightsquigarrow^1_{\alpha,d}$ and its transitive closure
$\rightsquigarrow_{\alpha,d}$.

\begin{definition} \label{operational}
Given a domain of actions $(\mathds{A},{\bf d})$ and a discount factor $\alpha\in (0,1]$, we inductively define the distance steps on $FTrees(\mathds{A})$ by
\vspace{0.15cm}

$1.$ $d\geq 0$ $\Rightarrow$ $(t\rightsquigarrow^1_{\alpha,d} t+t$ $\wedge$ $t+t\rightsquigarrow^1_{\alpha,d} t)$.
\hfill $2.$ $d\geq {\bf d}(a,b)$ $\Rightarrow$ $at\rightsquigarrow^1_{\alpha,d} bt$.\hspace{2.5cm}

$3.$ $t\rightsquigarrow^1_{\alpha,d} t^\prime$ $\Rightarrow$ $t+t^{\prime\prime}\rightsquigarrow^1_{\alpha,d} t^\prime+t^{\prime\prime}$.
\hfill$4.$ $t\rightsquigarrow^1_{\alpha,d} t^\prime$ $\Rightarrow$ $at\rightsquigarrow^1_{\alpha,\alpha d} at^\prime$.\hspace{2.5cm}

\vspace{0.15cm}
We associate to each distance step its level. The level of any step generated by $1.$ or $2.$ is one; while if the level of the corresponding premise $t\rightsquigarrow^1_{\alpha,d} t^\prime$ is $k$, then the level of a step generated by $3.$ (resp. $4.$) is $k$ (resp. $k+1$). 
Finally, we define the family of global distance relations $\langle\rightsquigarrow_{\alpha,d}\mid d \in \mathds{R}^+\rangle$, taking $t\rightsquigarrow_{\alpha,d} t^\prime$ if there exists a sequence $\mathcal{S}:=t=t^0\rightsquigarrow_{\alpha,d_1}^1 t^1 \rightsquigarrow_{\alpha,d_2}^1 t^2\rightsquigarrow_{\alpha,d_3}^1\dots \rightsquigarrow_{\alpha,d_n}^1 t^n=t^\prime,$ with $\sum_{i=1}^n d_i \leq d$.
\end{definition}

\section{The coinductive global bisimulation distance} \label{coinductive_section}
In order to extend our global distances to infinite trees, we have introduced in \cite{rf14} a general coinductive notion of distance. We formalize our definition in two steps. In the first one we introduce the rules that produce the steps of the \emph{coinductive transformations} between trees, starting from any family of triples $(t,t^\prime,d)$, with $t$, $t^\prime\in FyTrees (\mathds{A})$ and $d\in \mathds{R}^+$.
 
\begin{definition}[\cite{rf14}]\label{single_step}
Given a domain of actions $(\mathds{A},{\bf d})$, a discount factor $\alpha\in (0,1]$ and a family $\mathcal{D}\subseteq FyTrees(\mathds{A})\times FyTrees(\mathds{A})\times \mathds{R}^+$, we define the family of relations $\equiv_d^{\mathcal{D},\alpha}$, for $d \in \mathbb R^+$, by:
\begin{compactenum}
\item For all $d\geq 0$ we have (i) $(\sum_{j\in J} a_jt_j)+at+at \equiv_d^{\mathcal{D},\alpha} (\sum_{j\in J} a_jt_j)+at$ ,\\and (ii) $(\sum_{j\in J} a_jt_j)+at \equiv_d^{\mathcal{D},\alpha}(\sum_{j\in J} a_jt_j)+at+at$. 
\item For all $d\geq {\bf d}(a,b)$ we have $(\sum_{j\in J} a_jt_j)+at\equiv_{d}^{\mathcal{D},\alpha} (\sum_{j\in J} a_jt_j)+bt$.
\item For all $(t,t^\prime,d)\in \mathcal{D}$ we have $(\sum_{j\in J} a_jt_j)+at\equiv_{d^\prime}^{\mathcal{D},\alpha} (\sum_{j\in J} a_jt_j)+at^\prime$ for all $d^\prime\geq \alpha d$.
\end{compactenum}
\noindent To simplify the notation, we write $\equiv_d$ instead of $\equiv_d^{\mathcal{D},\alpha}$, whenever $\mathcal{D}$ and $\alpha$ are clear from the context.
\end{definition}

We say that the steps generated by application of rules $1$ and $2$ in Def.~\ref{single_step} are \emph{first level steps}; while those generated by rule $3$ are \emph{coinductive steps}.
Inspired by the proof obligations imposed to bisimulations we introduce the coinductive proof obligations imposed to the (satisfactory) families of distances. 

\begin{definition}[\cite{rf14}]\label{finitary_case}
Given a domain of actions $(\mathds{A},{\bf d})$ and a discount factor $\alpha\in (0,1]$, we say  that a family  $\mathcal{D}$ is an $\alpha$-coinductive collection of distances ($\alpha$-ccd) between finitary trees, if for all $(t,t^\prime,d)\in \mathcal{D}$ there exists a \emph{finite coinductive transformation sequence} $\mathcal{C}:=t=t^0\equiv_{d_1}t^1\equiv_{d_2}\ldots\equiv_{d_{n}}t^n=t^\prime$, with $d\geq\sum_{j=1}^n d_j$.
Then, when there exists an $\alpha$-ccd $\mathcal{D}$ with $(t,t^\prime,d)\in \mathcal{D}$, we will write $t\equiv_d^{\alpha}t^\prime$, and say that tree $t$ is at most at distance $d$ from tree $t^\prime$ wrt $\alpha$.
\end{definition} 

\begin{example} \label{ejem_infinito}
Let us consider the domain of actions $(\mathds{N}, {\bf d})$, where ${\bf d}$ is the usual distance for numbers, and the trees $t_N=unfold(N)$ and $t_{N^\prime}=unfold(N^\prime)$, with $N=\{n_0,n_1\}$, $\mathit{succ}(n_0)=\{(0,n_0),(0,n_1)\}$ and $\mathit{succ}(n_1)=\emptyset$; and $N^{\prime}=\{n_0^{\prime},n_1^{\prime}\}$, $\mathit{succ}^\prime(n_0^{\prime})=\{(0,n_0^{\prime}),(1,n_1^{\prime})\}$ and $\mathit{succ}^\prime(n_1^{\prime})=\emptyset$. Then, we have $t_{N}\equiv^{\mathcal{D},1/2}_2 t_{N^{\prime}}$, using the family $\mathcal{D}=\{(t_N,t_{N^{\prime}},2)\}$. We can prove that this is indeed a $\frac{1}{2}$-ccd, by considering the sequence: $\mathcal{C}:=t_{N}\equiv^{\mathcal{D},1/2}_1 t_{N^{\prime\prime}}\equiv^{\mathcal{D},1/2}_1 t_{N^{\prime}}$, where $t_{N^{\prime\prime}}=unfold(N^{\prime\prime})$, with $N^{\prime\prime}=\{n_0^{\prime\prime},n_1^{\prime\prime},n_2^{\prime\prime},n_3^{\prime\prime}\}$, $\mathit{succ}^{\prime\prime}(n_0^{\prime\prime})=\{(1,n_1^{\prime\prime}),$ $(0,n_2^{\prime\prime})\}$, $\mathit{succ}^{\prime\prime}(n_1^{\prime\prime})=\emptyset$, $\mathit{succ}^{\prime\prime}(n_2^{\prime\prime})=\{(0,n_2^{\prime\prime}),(0,n_3^{\prime\prime})\}$ and $\mathit{succ}^{\prime\prime}(n_3^{\prime\prime})=\emptyset$. 
\end{example}

It is immediate to see that our notion of distance has natural and pleasant properties such as the \emph{triangular transitivity}: for any discount factor $\alpha \in (0,1]$, whenever we have $t\equiv^{\alpha}_d t^\prime$ and $t^\prime\equiv^{\alpha}_{d^\prime} t^{\prime\prime}$, we also have $t\equiv^{\alpha}_{d+d^\prime}t^{\prime\prime}$.
Of course, our coinductive definition of the global bisimulation distance generalizes our operational definition for finite trees.

\begin{proposition} [\cite{rf14}]\label{coincidence_finite_coinductive}
For $t,t^\prime\in FTrees(\mathds{A})$, the operational (Def.~\ref{operational}) and the coinductive (Def.~\ref{finitary_case}) definition of global bisimulation distance between trees coincide, that means $t\equiv^{\mathcal{\alpha}}_d t^\prime \Leftrightarrow t\rightsquigarrow_{\alpha,d} t^\prime$.
\end{proposition}
\begin{proof}
$\underline{\Rightarrow|}$ We assume that $t\equiv^{\mathcal{\alpha}}_d t^\prime$
holds. This means that $(t,t^\prime,d)\in\mathcal{D}$ for some
$\alpha$-ccd $\mathcal{D}$, which in turn implies the existence of a finite
coinductive sequence
$\mathcal{C}:=t=t^0\equiv_{d_1}t^1\equiv_{d_2}\ldots\equiv_{d_{n}}t^n=t^\prime$,
with $\sum d_j \leq d$.
If $\mathcal{C}$ is the vacuous sequence ($n = 0$), then we have $t = t'$ and
$t\rightsquigarrow_{\alpha,d} t^\prime$ trivially holds.
Let us assume $n > 0$. We show that it is possible to ``unfold'' $\mathcal{C}$
into a sequence of distance steps, $\mathcal{S}$, proving that
$t\rightsquigarrow_{\alpha,d} t^\prime$.
It is enough to show that for a generic step $t^i\equiv_{d_{i+1}}t^{i+1}$ of
$\mathcal C$ there exists a sequence proving that $t^i
\rightsquigarrow_{\alpha,d_{i+1}} t^{i+1}$ (the complete sequence for $\mathcal C$
can be obtained by composition).
If $t^i\equiv_{d_{i+1}}t^{i+1}$ has been obtained by applying rule $1$ in
Def.~\ref{single_step} (we only consider the sub-case (i), the other case can
be dealt with in the same way), then we have
$t^i = t^i_1 + at_1 + at_1 \equiv_{d_{i+1}} t^i_1 + at_1 = t^{i+1}$, and applying
rules 1 and 3 in Def.~\ref{operational} we get
$t^i\rightsquigarrow_{\alpha,d_{i+1}} t^{i+1}$.
If $t^i\equiv_{d_{i+1}} t^{i+1}$ has been obtained by applying rule $2$ in
Def.~\ref{single_step}, then we have
$t^i = t^i_1 + at_1 \equiv_{d_{i+1}} t^i_1 + bt_1 = t^{i+1}$, and applying
rules 2 and 3 in Def.~\ref{operational} we get
$t^i\rightsquigarrow_{\alpha,d_{i+1}} t^{i+1}$.
If $t^i\equiv_{d_{i+1}} t^{i+1}$ has been obtained by applying rule $3$ in
Def.~\ref{single_step}, then we have
$t^i = t^i_1 + at_1 \equiv_{d_{i+1}} t^i_1 + at'_1 = t^{i+1}$, with $(t_1,t'_1,d')$
for some $d'$ such that $d_{i+1} \geq \alpha d'$.
We proceed by induction on the depth of $t^i$.
Notice that $depth(t^i)>0$ as there is no $t''$ such that
$t^i\equiv_{d_{i+1}} t''$ when $depth(t^i)=0$.
If $depth(t^i)=1$ (base case), then $t_1 = {\bf 0}$, and the only possible
witness for $(t_1,t'_1,d')$ is the vacuous coinductive sequence.
Thus, we have $t_1 = t'_1 = {\bf 0}$, and
$t_1\rightsquigarrow_{\alpha,d_{i+1}/\alpha} t'_1$ trivially holds.
By applying rules 4 and 3 in Def.~\ref{operational} we get
$t^i\rightsquigarrow_{\alpha,d_{i+1}} t^{i+1}$.
Finally, if $depth(t^i)>1$, then $t_1\rightsquigarrow_{\alpha,d'} t'_1$ holds by
inductive hypothesis, and applying rules 4 and 3 in Def.~\ref{operational} we
get $t^i\rightsquigarrow_{\alpha,\alpha d'} t^{i+1}$, which in turn, trivially
implies $t^i\rightsquigarrow_{\alpha,d_{i+1}} t^{i+1}$ (as $d_{i+1} \geq \alpha d'$).

$\underline{\Leftarrow|}$ Given a sequence of distance steps, $\mathcal{S}$, proving that $t\rightsquigarrow_{\alpha,d} t^\prime$, we can ``fold'' it into a coinductive sequence, $\mathcal{C}$, witnessing $t\equiv^{\alpha}_d t^\prime$. For each $(t,t^\prime,d)\in\mathcal{D}$ we consider the factorization of the sequence $\mathcal{S}$ and its reordering as done in \cite{rf14}. We get $t=\sum_{i\in I_0} a_it_i\rightsquigarrow_{\alpha,d_{02}} \sum_{i\in I_0}a_it^\prime_i\rightsquigarrow^1_{\alpha,d_{11}}\sum_{i\in I_1}a_it_i\rightsquigarrow_{\alpha,d_{12}}\sum_{i\in I_1} a_i t^\prime_i\rightsquigarrow^1_{\alpha,d_{21}} \dots \rightsquigarrow_{\alpha,d_{(k+1)2}} \sum_{i\in I_{k+1}}a_i t^\prime_i=t^\prime$, where for each sequence $\sum_{i\in I_j} a_it_i\rightsquigarrow_{\alpha,d_{j2}}\sum_{i\in I_j}a_it^\prime_i$ and each $i\in I_j$, we have $t_i\rightsquigarrow_{\alpha,d_{j2}^i/\alpha}t^\prime_i$ with $\sum_{i\in I_j} d_{j2}^i=d_{j2}$. Now, applying the induction hypothesis, we have $(t_i,t^\prime_i, d_{j2}^i/\alpha)\in \mathcal{D}$, for all $i\in I_j$, so that $\sum a_it_i\equiv^{\mathcal{D},\alpha}_{\alpha d_1} \sum a_it_i^1\equiv^{\mathcal{D},\alpha}_{\alpha d_2}\sum a_it^2_i\equiv^{\mathcal{D},\alpha}_{\alpha d_3}\dots \equiv^{\mathcal{D},\alpha}_{\alpha d_{|I_j|}}\sum a_it_i^{|I_j|}=\sum a_it^\prime_i$.

Therefore, each sequence $\sum_{i\in I_j} a_it_i\rightsquigarrow_{\alpha,d_{j2}^i}\sum_{i\in I_j}a_it^\prime_i$ at the factorization above can be substituted by a sequence of $|I_j|$ valid coinductive steps, getting a total distance $\sum_{j=0}^{k+1}\sum_{i\in I_j} d_{j1}^i + \sum_{j=0}^{k+1}\sum_{k=1}^{|I_j|} \alpha d_k=d$. 
\end{proof}

Even if the result above only concerns finite trees, it reveals the duality between induction and coinduction. Of course, its consequences are much more interesting in the infinite case. 

\begin{example}\label{ainfinito_binfinito}
Let $(\{ a,b \}, {\bf d})$ be a domain of actions such that ${\bf d}(a,b) = 1$
and let us consider the lts given by $N_{1,\infty}=\{n_0\}$, with $\mathit{succ}(n_0)=\{(a,n_0)\}$ and its unfolding $a^\infty$. In an analogous way, we obtain the tree $b^\infty$. We have $\pi_n(a^\infty)=a^n$, and clearly $a^\infty$ can be seen as the ``limit'' of its projections. Applying Def.~\ref{operational}, we obtain $a^n\rightsquigarrow_{1/2,2} b^n$ and thus $a^n\equiv^{1/2}_2 b^n$. We also have $a^\infty\equiv^{1/2}_2 b^\infty$, which can indeed be proved by means of the (trivial!) collection $\mathcal{D}=\{(a^\infty,b^\infty,2)\}$. We can check that $\mathcal{D}$ is an $\frac{1}{2}$-ccd using the coinductive sequence $\mathcal{C}:=a^\infty=aa^\infty\equiv^{\mathcal{D},1/2}_1ba^\infty\equiv^{\mathcal{D},1/2}_{\frac{1}{2}\cdot 2} bb^\infty=b^\infty$.
\end{example}

\section{On the continuity of the global bisimulation distance}
It is very simple to prove the following proposition, by introducing the notion of projections of $\alpha$-ccd's.
\begin{proposition}[\cite{rf14}]\label{ccd_projection}
For any $\alpha$-ccd $\mathcal{D}$, the projected family $\pi(\mathcal{D})=\{(\pi_n(t),\pi_n(t^\prime),d)\mid (t,t^\prime,d)\in \mathcal{D},\; n\in \mathds{N}\}$ is an $\alpha$-ccd. Hence, it holds $t \equiv^{\alpha}_d t^\prime \Rightarrow \pi_n(t)\equiv^{\alpha}_d \pi_n(t^\prime)$ for each
$n\in \mathds{N} $.
\end{proposition}
\begin{proof}
Let $\mathcal{C}:=t=t^0\equiv^{\mathcal{D},\alpha}_{d_1}\ldots\equiv^{\mathcal{D},\alpha}_{d_{k}} t^k=t^\prime$ be the coinductive sequence proving that $(t,t^\prime,d)\in \mathcal{D}$ satisfies the condition in order $\mathcal{D}$ to be an $\alpha$-ccd. Then each projected sequence $\pi_n(\mathcal{C}):=\pi_n(t)=\pi_n(t^0)\equiv^{\pi(\mathcal{D}),\alpha}_{d_1} \ldots \equiv^{\pi(\mathcal{D}),\alpha}_{d_{k}}\pi_n(t^k)=\pi_n(t^\prime)$ proves that $(\pi_n(t),\pi_n(t^\prime),d)\in\pi(\mathcal{D})$ satisfies the condition in order $\pi(\mathcal{D})$ to be an $\alpha$-ccd. It is clear that the projection under $\pi_n$ of any first level step in $\mathcal{C}$, is also a valid step in $\pi_n(\mathcal{C})$. Moreover, any coinductive step in $\mathcal{C}$ using $(t_1, t^{\prime}_1, d')\in \mathcal{D}$, can be substituted by the corresponding projected step, that uses $(\pi_{n-1}(t_1), \pi_{n-1}(t^{\prime}_1),d')\in \pi(\mathcal{D})$.
\end{proof}

\begin{remark}\label{projection_family}
Alternatively, we can consider for each $n\in\mathds{N}$ a family $\mathcal{D}_n=\pi_n(\mathcal{D})=\{(\pi_m(t),\pi_m(t^\prime),d)\mid (t,t^\prime,d)\in \mathcal{D},\; m\in \mathds{N}, m\leq n\}$, using the fact that the subtrees of a projection $\pi_n(t)$ are also projections $\pi_m(t^{\prime\prime})$ of subtrees $t^{\prime\prime}$ of $t$, for some $m<n$. These families satisfy $\pi_m(\mathcal{D})\subseteq \pi_n(\mathcal{D})$, whenever $m\leq n$.
\end{remark}

We conjecture that the converse of Prop.~\ref{ccd_projection}, asserting the continuity of our coinductive distance, is also valid. However, after several attempts have failed, or led us into a blind alley, we are now less confident of this result than we were before. We expected to be able to prove it by using the argument used in the proof of Prop.~\ref{ccd_projection}, in the opposite direction. If we would have a collection of ``uniform''\footnote{Uniformity here means that for any $n,k\in\mathds{N}$, with $k\geq n$, the steps corresponding to the first $n$ levels of any sequence $\mathcal{S}^k$ are always the same.} operational sequences $\mathcal{S}^n:=\pi_n(t)\rightsquigarrow_{\alpha,d}\pi_n(t^\prime)$, we could ``overlap'' all of them getting an ``infinite tree-structured sequence''. By ``folding'' it we would obtain the coinductive sequence proving that $t\equiv^{\alpha}_dt^\prime$. Next, a simple example.

\begin{example}[\cite{rf14}]\label{another_Ex}
Let us consider the trees $t=ac^\infty+ad^\infty$ and $t^\prime=bc^\infty+bd^\infty$, and the distance ${\bf d}$ defined by ${\bf d}(a,b)=4$, ${\bf d}(c,d)=1$. Then we have (the numbers above the arrows denote the distance step level, according to Def.~\ref{operational}; moreover, even if the arrows denote here 1-step transformations, we omit the superscript 1 for the sake of readability):
\vspace{-0.15cm} 
\begin{center}
$\pi_1(t)=a+a\stackrel{\mbox{\tiny (1)}}{\rightsquigarrow}_{\frac{1}{2},0}a\stackrel{\mbox{\tiny (1)}}{\rightsquigarrow}_{\frac{1}{2},4}b\stackrel{\mbox{\tiny (1)}}{\rightsquigarrow}_{\frac{1}{2},0}b+b=\pi_1(t^\prime)$, 

$\pi_2(t)=ac+ad\stackrel{\mbox{\tiny (2)}}{\rightsquigarrow}_{\frac{1}{2},\frac{1}{2} \cdot 1}ac+ac\stackrel{\mbox{\tiny (1)}}{\rightsquigarrow}_{\frac{1}{2},0}ac\stackrel{\mbox{\tiny (1)}}{\rightsquigarrow}_{\frac{1}{2},4}bc\stackrel{\mbox{\tiny (1)}}{\rightsquigarrow}_{\frac{1}{2},0}bc+bc\stackrel{\mbox{\tiny (2)}}{\rightsquigarrow}_{\frac{1}{2},\frac{1}{2} \cdot 1}bc+bd=\pi_2(t^\prime)$, 

$\pi_3(t)=acc+add\stackrel{\mbox{\tiny (2)}}{\rightsquigarrow}_{\frac{1}{2},\frac{1}{2}\cdot 1} acc+acd\stackrel{\mbox{\tiny (3)}}{\rightsquigarrow}_{\frac{1}{2},\frac{1}{4}\cdot 1} acc+acc\stackrel{\mbox{\tiny (1)}}{\rightsquigarrow}_{\frac{1}{2},0} acc\stackrel{\mbox{\tiny (1)}}{\rightsquigarrow}_{\frac{1}{2},4}$

\hspace{1.5cm}$bcc\stackrel{\mbox{\tiny (1)}}{\rightsquigarrow}_{\frac{1}{2},0}bcc+bcc\stackrel{\mbox{\tiny (2)}}{\rightsquigarrow}_{\frac{1}{2},\frac{1}{2}\cdot 1} bcc+bdc\stackrel{\mbox{\tiny (3)}}{\rightsquigarrow}_{\frac{1}{2},\frac{1}{4}\cdot 1}bcc+bdd=\pi_3(t^\prime)$. 
\end{center}

Now, if we consider the operational sequences, $\mathcal{S}^n$, relating $\pi_n(t)$ and $\pi_n(t^\prime)$, for any $n\in\mathds{N}$, we obtain $\pi_n(t)\rightsquigarrow_{\frac{1}{2},d_n}\pi_n(t^\prime)$, for some $d_n<6$. 
Applying Prop.~\ref{coincidence_finite_coinductive} we can turn these operational sequences into equivalent coinductive ones, $\mathcal{C}^n$, thus proving $\pi_n(t)\equiv^{1/2}_{6}\pi_n(t^\prime)$, for all $n\in \mathds{N}$:
\vspace{-0.2cm}
\begin{center}
$\mathcal{C}^1:=a+a\equiv_{0}a\equiv_{4}b\equiv_{0}b+b$, 

$\mathcal{C}^2:=ac+ad\equiv_{\frac{1}{2}}ac+ac\equiv_{0}ac\equiv_{4}bc\equiv_{0}bc+bc\equiv_{\frac{1}{2}}bc+bd$,

$\mathcal{C}^3:=acc+add\equiv_{\frac{1}{2}}acc+acd\equiv_{\frac{1}{4}}acc+acc\equiv_{0}acc\equiv_{4}bcc\equiv_{0}bcc+bcc\equiv_{\frac{1}{2}}bcc+bdc\equiv_{\frac{1}{4}}bcc+bdd$.
\end{center}
\end{example}

\begin{lemma}\label{borrar_niveles}
If $t\rightsquigarrow^1_{\alpha,d} t^\prime$ is a distance step of level $l$, then $\pi_k(t)\rightsquigarrow^1_{\alpha,d} \pi_k(t^\prime)$ is also a distance step of level $l$ for all $k\geq l$ .
Instead, if $l > k$ then we have $\pi_k(t)=\pi_k(t^\prime)$.
\end{lemma}

The continuity of our coinductive distance would mean that, whenever we have $\pi_n(t)\equiv^\alpha_d\pi_n(t^\prime)$ $\forall n\in\mathds{N}$, there should be a collection of ``uniform'' sequences proving these facts. Certainly, this is the case when we know in advance that $t\equiv^\alpha_d t^\prime$. Next, we define the projection of our operational sequences.

\begin{definition}
Given a sequence of distance steps $\mathcal{S}:=t=t^0\rightsquigarrow^1_{\alpha,d_1} t^1 \rightsquigarrow^1_{\alpha,d_2}\dots\rightsquigarrow^1_{\alpha,d_n}t^n=t^\prime$ with $d=\sum_{i=1}^n d_i$, for any $k\in \mathds{N}$ we define the sequence $\pi_k(\mathcal{S}):=\pi_k(t)=\pi_k(t^0)\rightsquigarrow^1_{\alpha,d_{i_1}} \pi_k(t^{i_1}) \rightsquigarrow^1_{\alpha,d_{i_2}} \dots \rightsquigarrow^1_{\alpha,d_{i_{l}}} \pi_k(t^{i_l})=\pi_k(t^\prime)$, where $\langle i_1,i_2,\dots,i_l\rangle$ is the sequence of indexes $i_j$ for which the level of the distance step $t^{i_j-1}\rightsquigarrow^1_{\alpha,d_{i_{j}}}t^{i_{j}}$ is less than or equal to $k$, while the rest of the steps are removed.
\end{definition}

\begin{proposition}
For any $k\in\mathds{N}$, and any sequence of distance steps $\mathcal{S}$, the projected sequence $\pi_k(\mathcal{S})$ is a sequence of distance steps, thus proving $\pi_k(t)\rightsquigarrow_{\alpha,\sum d_{i_j}} \pi_k(t^\prime)$. Therefore, we also have $\pi_k(t)\rightsquigarrow_{\alpha,d}\pi_k(t^\prime)$.
\end{proposition}

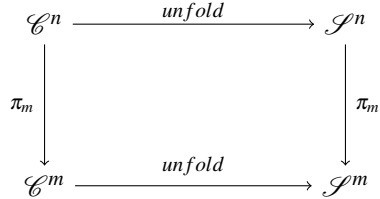
\begin{figure}
\begin{center}
\begin{tikzpicture}
\node(Cn) at (-2,1.2){$\mathcal{C}^n$};
\node(Sn) at (2,1.2){$\mathcal{S}^n$};
\node(Cm) at (-2,-0.95){$\mathcal{C}^m$};
\node(Sm) at (2,-0.95){$\mathcal{S}^m$};

\begin{scope}[every node/.style={font=\scriptsize\itshape}]
\path [->]   
            (Cn) edge node[auto,swap]{$\pi_m$}(Cm)
            (Sn) edge node[auto]{$\pi_m$}(Sm)
            (Cn) edge node[above=-1mm]{$unfold$} (Sn)
            (Cm) edge node[auto]{$unfold$} (Sm)
            ;
\end{scope}
\end{tikzpicture}
\end{center}
\vspace{-0.5cm}
\caption{Relating projections of sequences}\label{diagrama}
\end{figure}

Now, let us start with a coinductive sequence $\mathcal{C}$ relating two (possibly infinite but finitary) trees $t$ and $t^\prime$.

\begin{corollary}
The projection of coinductive sequences and those of distance steps that the former induces, are related by the commutative diagram in Fig.~\ref{diagrama}. There we denote by $\mathcal{C}^n$ the projections of a coinductive sequence $\mathcal{C}$.
\end{corollary}

Note that all these coinductive sequences have exactly the ``same structure'', which is formalized by the fact that $\pi_m(\mathcal{C}^n)=\mathcal{C}^m$, for all $m\leq n$. As a consequence, if now we forget that we have already the original family $\mathcal{D}$ in Ex.~\ref{ainfinito_binfinito}, starting from the families $\mathcal{D}_n$ in Remark \ref{projection_family}, we could ``reverse'' the procedure by means of which they were defined, obtaining a single family $\mathcal{D}^\prime=\bigcup \mathcal{D}_n\cup\{(a^\infty,b^\infty,2)\}=\{(a^k,b^k,2)\mid k\leq \infty\}$, where we have added the ``limit'' triple $(a^\infty, b^\infty,2)$ because for all $k\in\mathds{N}$ $(\pi_k(a^\infty),\pi_k(b^\infty),2)=(a^k,b^k,2)\in \bigcup\mathcal{D}_n$. Now, we can see that $\mathcal{D}^\prime$ is indeed a $\frac{1}{2}$-ccd. In order to check the condition corresponding to $(a^\infty,b^\infty,2)$, we ``overlap'' the sequences $\mathcal{C}^n$ getting a sequence $\mathcal{C}^\prime$ constructed as follows: 
From the first level step $a^n=aa^{n-1}\equiv^{\mathcal{D}_n,1/2}_{1}ba^{n-1}$ at each $\mathcal{C}^n$, we obtain the first level step $a^\infty=aa^\infty\equiv^{\mathcal{D}^\prime,1/2}_{1}ba^\infty$; this is followed by the coinductive step $ba^\infty\equiv^{\mathcal{D}^\prime,1/2}_{\frac{1}{2} \cdot 2}bb^\infty=b^\infty$, obtained just removing the projections from any of the coinductive steps $ba^{n-1}=\pi_n(ba^{\infty})\equiv^{\mathcal{D}_n,1/2}_{\frac{1}{2}\cdot 2}\pi_n(bb^\infty)=\pi_n(b^\infty)=b^n$.

The important fact about the construction above is that it can be applied to any collection of coinductive sequences that ``match each other''. We will call these collections \emph{telescopic}, because when we ``unfold'' their elements we obtain a sequence of operational sequences, where any of them is obtained from the previous one by adding some new intermediate steps, that always correspond to transformation steps at the last level of the compared trees (this reminds us the ``opening'' of a telescopic antenna).

\begin{definition}
Let $t,t^\prime\in FyTrees(\mathds{A})$ and let $(\mathcal{S}^n)_{n \in \mathds{N}}$ be a collection of operational sequences proving $\pi_n(t)\rightsquigarrow_{\alpha,d}\pi_n(t^\prime)$. We say that it is \emph{telescopic}\footnote{Certainly, these ``telescopic'' sequences correspond to the notion of inverse limit in domain theory or category theory. But since we only need a very concrete case of that quite abstract notion, we prefer to define it in an explicit way here.} if $\pi_{m}(\mathcal{S}^n)=\mathcal{S}^m$, for all $m,n \in \mathbb N$, with $m\leq n$.
\end{definition}

Any telescopic collection $(\mathcal{S}^n)_{n \in \mathds{N}}$, produces a ``limit'' coinductive sequence $\mathcal{C}$ proving $t\equiv^\alpha_d t^\prime$:

\begin{lemma}\label{lemita}
Let $(\mathcal{S}^n)_{n\in\mathds{N}}$ be a telescopic collection relating $t$ and $t^\prime$. We consider the associated factorization of each sequence in it: $\mathcal{S}^n:=\pi_n(t)=t^{n,0,1}\rightsquigarrow_{\alpha,d_{n02}} t^{n,0,2}\rightsquigarrow^1_{\alpha,d_{n11}}t^{n,1,1}\dots\rightsquigarrow^1_{\alpha,d_{nk1}}t^{n,k,1}\rightsquigarrow_{\alpha,d_{nk2}} t^{n,k,2}=\pi_n(t^\prime)$ where we alternate distance steps at the first level and global steps, that aggregate a sequence of steps at deeper levels. Then, for all $m\leq n$, $j\in\{1,\dots,k\}$ and $r\in\{1,2\}$, we have $t^{m,j,r}=\pi_m(t^{n,j,r})$.
\end{lemma}
\begin{proof}
Immediate, by definition of projections and telescopic collections.
\end{proof}

\begin{corollary}\label{cor}
For all $n\in\mathds{N}$, $j\in\{1, \dots, k\}$ and $r\in\{1,2\}$, we can decompose $t^{n,j,r}$ as above into $\sum_{i=1}^{I_{njr}} a_i t^{n,j,r}_i$, where the sequences $t^{n,j,1}\rightsquigarrow_{\alpha,d_{nj2}}t^{n,j,2}$ satisfy $I_{nj2}=I_{nj1}$ $\forall j\in\{1\dots k\}$,
and can be factorized into: $t^{n,j,1}=\sum a_it_i^{n,j,1}\rightsquigarrow_{\alpha,d_{nj2,1}}\sum a_it_i^{n,j,2,1}\dots\rightsquigarrow_{\alpha,d_{nj2,I_{nj2}}}\sum a_it_i^{n,j,2,I_{nj2}}=$ $\sum a_it^{n,j,2}_i=t^{n,j,2}$, where $t_i^{n,j,2,l}=t_i^{n,j,2}$ $\forall i\leq l$ and $t_i^{n,j,2,l}=t_i^{n,j,1}$ $\forall i>l$, and $\sum_{i=1}^{I_{nj2}} d_{nj2,i}=d_{nj2}$.

As a consequence, if we denote by $t^{\infty,j,r}$ the unique tree in $FyTree$ that satisfies $\pi_n(t^{\infty,j,r})=t^{n,j,r}$ $\forall n\in\mathds{N}$, we can decompose it into $\sum_{i=1}^{I_{j,r}}a_it_i^{\infty,j,r}$, and each collection $(\mathcal{S}^n_{j,i})_{n\in\mathds{N}}$ given by $\mathcal{S}^n_{j,i}:=\pi_n(t_i^{\infty,j,1})=t_i^{n,j,1}\rightsquigarrow_{\alpha,d_{nj2}}t_i^{n,j,2}=\pi_n(t^{\infty,j,2}_i)$ is indeed a telescopic collection relating $t_i^{\infty,j,1}$ and $t_i^{\infty,j,2}$.
\end{corollary}
\begin{proof}
Easy to check, by definition of projections and telescopic collections.
\end{proof}

\begin{theorem}\label{proyecciones}
Let $\mathcal{D}=\{(t,t^\prime,d)\mid \exists$ $(\mathcal{S}^n)_{n\in \mathds{N}}$ a telescopic collection relating $t$ and $t^\prime\}$, then $\mathcal{D}$ is an $\alpha$-ccd.
\end{theorem}
\begin{proof}
Using the notation in Lem.~\ref{lemita} and Cor.~\ref{cor} above, we can consider the coinductive sequence
\begin{center}
$t=t^{\infty,0,1}\equiv^{\mathcal{D},\alpha}_{d_{02}} t^{\infty,0,2}\equiv^{\mathcal{D},\alpha}_{d_{11}}t^{\infty,1,1}\equiv^{\mathcal{D},\alpha}_{d_{12}}\dots\equiv^{\mathcal{D},\alpha}_{d_{k2}}t^{\infty,k,2}=t^\prime,$
\end{center}
 where $d_{j2}=sup_{n\in\mathds{N}}\{d_{nj2}\}$ and $d_{j1}=d_{nj1}$, $\forall n\in \mathds{N}$. It is clear that all the steps $t^{\infty,j,2}\equiv^{\mathcal{D},\alpha}_{d_{(j+1)1}}t^{\infty,j+1,1}$ are valid first level steps, while using Cor.~\ref{cor} we have that the steps $t^{\infty,j,1}\equiv^{\mathcal{D},\alpha}_{d_{j2}}t^{\infty,j,2}$ correspond to valid coinductive steps from $\mathcal{D}$.  This is so because joining all the members of the telescopic sequences $(\mathcal{S}^n_{j,i})_{n\in\mathds{N}}$ with $j\in\{0\dots k\}$ fixed, we obtain a single telescopic sequence $(\mathcal{S}^n_{j})_{n\in\mathds{N}}$ relating $t^{\infty,j,1}=\sum_{i=1}^{I_{j,1}}a_it_i^{\infty,j,1}$ and $t^{\infty,j,2}=\sum_{i=1}^{I_{j,2}}a_it_i^{\infty,j,2}$.
\end{proof}

If for any $t,t^\prime\in FyTrees$
there would be only finitely many sequences proving each valid triple $t\rightsquigarrow_{\alpha,d}t^\prime$, then a classical compactness technique (or K\"onig's lemma, if you prefer) would immediately prove that whenever we have $\pi_n(t)\equiv^{\alpha}_{d}\pi_n(t^\prime)$ for any $n\in\mathds{N}$, we can obtain a telescopic collection of sequences proving all these facts. Then, the application of Th.~\ref{proyecciones} would conclude the continuity of our global bisimulation distance.
Unfortunately, this is not the case, because we can arbitrary enlarge any such sequence, adding dummy steps that apply the idempotency rule in one and the other directions. Even more, in some complicated cases in order to obtain some distances between two trees we need to consider sequences that include intermediate trees wider than the compared ones. 

\begin{example}\label{ej1}
Let $\mathds{A}=\{1,2,3,4,5\}$ with the ``usual'' distance ${\bf d}(n,m)=|m-n|$, for all $m,n\in\mathds{A}$. Let us consider the trees $t=1(2+3+4+5)+1(1+2+3+4)$ and $t^\prime=1(1+2+4+5)$. Then, we have $t\rightsquigarrow_{1,3} t^\prime$, which can be obtained by means of the sequence $\mathcal{S}:=t\rightsquigarrow^{1}_{1,0}1(2+2+3+4+5)+1(1+2+3+4)\rightsquigarrow^{1}_{1,1} 1(1+2+3+4+5)+1(1+2+3+4)\rightsquigarrow^{1}_{1,0} 1(1+2+3+4+5)+1(1+2+3+4+4)\rightsquigarrow^{1}_{1,1} 1(1+2+3+4+5)+1(1+2+3+4+5)\rightsquigarrow^{1}_{1,0} 1(1+2+3+4+5)\rightsquigarrow^{1}_{1,1} 1(1+2+4+4+5)\rightsquigarrow^{1}_{1,0}1(1+2+4+5)=t^\prime$.
\end{example}

\begin{proposition}\label{prop_anchuras}
There exist $t,t^\prime\in FTrees(\mathds{A})$, with $t\rightsquigarrow_{\alpha,d} t^\prime$ and $||t||_m$, $||t^\prime||_m \leq k$ for some $m,k\in\mathds{N}$, for which it is not possible to obtain an operational sequence $\mathcal{S}$ witnessing $t\rightsquigarrow_{\alpha,d} t^\prime$ whose intermediate trees  $t^{i}$ satisfy $||t^i||_m\leq k$.
\end{proposition}
\begin{proof}
For the trees $t$, $t^\prime$ in Ex.~\ref{ej1}, there is no sequence $\mathcal{S}$ proving $t\rightsquigarrow_{1,3}^{1} t^\prime$ that only includes intermediate trees $t^{\prime\prime}$ with $||t^{\prime\prime}||_2\leq 4\,$.

We have to transform the sets $A=\{2,3,4,5\}$ and $B=\{1,2,3,4\}$ into $C=\{1,2,4,5\}$. First of all, we need to unify $A$ and $B$, otherwise if we try to get $C$ from $A$ and $B$ independently it will cost more than $3$ units. In order to unify $A$ and $B$, we need at least a $2$ units payment getting one of the following intermediate sets: $C_1^{\prime}=\{1,2,3,4,5\}$, $C_2^{\prime}=\{1,2,3,4\}$, $C_3^{\prime}=\{2,3,4,5\}$ or $C_4^{\prime}=\{2,3,4\}$. Now, we see that $C_1^\prime$ is the only intermediate set that gives us the $3$ units cost given in Ex.~\ref{ej1}.

\noindent\underline {$C_2^\prime\rightsquigarrow_{1,2}C:$} Pay $1$ unit to add $5$, and another to ``erase'' $3$, so that the total cost would be $4$.

\noindent\underline {$C_3^\prime\rightsquigarrow_{1,2}C:$} Pay $1$ unit to add $1$, and another to ``erase'' $3$, so that the total cost would be $4$ again. 

\noindent\underline {$C_4^\prime\rightsquigarrow_{1,3}C:$} Pay $1$ unit to add $1$, another to add $5$, and another to ``erase'' $3$; the total cost would be $5$.\end{proof}

Also, the following example shows us that in some cases the sequences at the telescopic sequences could be a bit involved. Even if for small values of $m$ the projections $\pi_m(t)$ and $\pi_m(t^\prime)$ could be ``quite close'', we need more elaborated sequences for relating them. Otherwise we would not be able to expand those sequences into others connecting $\pi_n(t)$ and $\pi_n(t^\prime)$, for larger values of $n$.

\begin{example}
We can check $ac+bd\rightsquigarrow_{1,2} ad+bc$ by using the sequence $\mathcal{S}:= ac+bd\rightsquigarrow^1_{1,1}bc+bd\rightsquigarrow^{1}_{1,1}bc+ad=ad+bc$, which corresponds to $\pi_1(\mathcal{S})=\mathcal{S}^1:= a$$+b\rightsquigarrow^1_{1,1}b$$+b\rightsquigarrow^1_{1,1} b$$+a=a+b$. Therefore, in this case, we cannot start with the trivial sequence $\mathcal{S}^{{\prime}1}:= a$$+b=b$$+a$, and moreover we need to take the order of ``summands'' into account, expressing the fact that we have really to change $a$ into $b$, and $b$ into $a$. 
\end{example}

\section{Some partial results looking for continuity}
Ex.~\ref{ej1} shows us that sometimes we need to use intermediate trees containing subtrees that are wider than the corresponding ones in the two compared trees. However, we expected to prove that the width of those subtrees would be bounded by some (simple) function of the width of the subtrees at the same depth of the compared trees. We tried a proof by induction on the depth of the trees, whose base case should correspond to depth $1$.

We consider $1$-depth trees $t=\sum_{i=1}^n a_i$ and $t^\prime=\sum_{j=1}^m b_j$,
and the corresponding multisets of labels $A=\{a_1,\dots, a_n\}$ and $B=\{b_1,\dots, b_m\}$. Any step in a sequence $t\rightsquigarrow_{\alpha,d}t^\prime$ is, of course, a first level step and corresponds to the application of either idempotency, in any direction, or the relabeling rule. We can represent these sequences by means of a multi-stage graph. We formally denote these graphs by $(S,T,l,v)$, where
\begin{inparaenum}[$(i)$]
\item $S$ is partitioned into $\{S_j\}_{j=0, \ldots, m}$,
\item $l\mid_{S_0}:S_0\rightarrow A$ and $l\mid_{S_m}:S_m\rightarrow B$ are bijections, and
\item arcs in $T$ are of the form $(t,u)$, with $t\in S_i$ and $u\in S_{i+1}$ for some $i$, and correspond to any of the patterns in Fig.~\ref{Pasos_nodos}, with the application of a single non-identity pattern at each stage.
\end{inparaenum}
Then, the full cost $d$ of the sequence is just the sum of the costs of all the arcs in the graph.

\begin{definition}
We say that a multi-stage graph is
\begin{compactenum}
\item \emph{Totally sides connecting} (tsc) if for all $s_0\in S_0$ there exists a path connecting it with some $s_m\in S_m$, and, symmetrically, for all $s_m\in S_m$ there exists a path connecting it with some $s_0\in S_0$.
\item \emph{Totally both ways connected} (tbwc) if it is tsc and for each $j\in \{1, \ldots, m-1\}$ and each $s\in S_j$ there exist two arcs $(s^\prime,s), (s,s^{\prime\prime})\in T$.
\end{compactenum}
\end{definition}

\begin{figure}[t]
\begin{center}
\scriptsize{
\begin{tikzpicture}[->,>=stealth',shorten >=1pt,auto]
\matrix [matrix of math nodes, column sep={1cm,between origins},row sep={0.5cm,between origins}]
{
\node(){};&\node (2) {c}; & \node (5) {c}; & \node () {};&\node(){};&\node(){};&\node(){};&\node(){};&\node(){};&\node(){};\\
\node(1){c};&\node () {}; & \node () {}; & \node (4) {c};&\node(){,};&\node(7){c};&\node(8){d};&\node(){,};&\node(9){c};&\node(10){c};\\
\node(){};&\node(3){c}; & \node (6) {c};& \node () {}; &\node(){};&\node(){};&\node(){};&\node(){};&\node(){};&\node(){};\\
\node(){}; & \node () {\hspace{1.1cm}Idempotency}; & \node () {};&\node(){};&\node(){};&\node(){};&\node(){\hspace{-.9cm}Relabeling};&\node(){};&\node(){};&\node(){};&\node(){\hspace{-3.3cm}Identity};\\
};

\begin{scope}[every node/.style={font=\scriptsize\itshape}]
\path [->]   
            (1) edge node[auto]{$0$}(2)
            (1) edge node[auto,swap]{$0$}(3)
            (5) edge node[auto]{$0$}(4)
            (6) edge node[auto,swap]{$0$}(4)
            (7) edge node[auto]{${\bf d}(c,d)$}(8)
            (9) edge node[auto]{0}(10);
\end{scope}
\end{tikzpicture}}
\end{center}
\vspace{-0.3cm}
\caption{Arcs of the graph representing the transformation of single level trees.}\label{Pasos_nodos}
\end{figure}
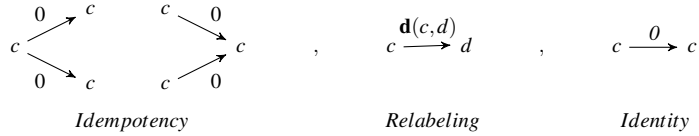

\begin{proposition}\label{cota_primerpiso}
Let $t=\sum_{i=1}^n a_i$ and $t^\prime=\sum_{j=1}^m b_j$ as above. Then, whenever we have $t\rightsquigarrow_{\alpha,d}t^\prime$, we can prove this by means of a sequence that has at most $3(m+n-2)+1$ distance steps. Consequently, the widths of the trees along such a sequence are also bounded by that amount.
\end{proposition}
\begin{proof}
We consider the multi-stage graph $\mathcal{G}$ that represents the sequence $t\rightsquigarrow_{\alpha,d}t^\prime$. We observe that this graph is tsc, using this fact we obtain a ``compacted'' subgraph $\mathcal{H}$ (of the original graph $\mathcal{G}$) by selecting a subset of the nodes of $\mathcal{G}$ and disjoints paths connecting them. So that, ever node in $\mathcal{G}$-$\mathcal{H}$ is at most in one of theses paths.

We turn these paths into arcs of $\mathcal{H}$ whose cost is $d(c_i,c_j)$, where $c_i$, $c_j$ are the two extremes of the path. By applying the triangle inequality, we immediately obtain that this cost is no larger than the cost of the original path. Therefore, the full cost of $\mathcal{H}$ is no bigger than that of $\mathcal{G}$. The set of nodes in $\mathcal{H}$ is obtained in the following way:
\begin{compactitem}
\item First, we consider $a_1\in A$ and some reachable $b_{j1}\in B$ from it. We introduce both of them in $\mathcal{H}$, and also add the arc connecting them, as explained above.
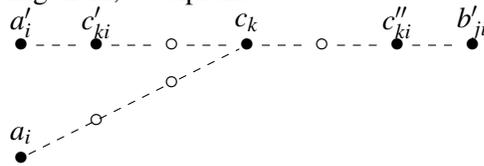
\vspace{-0.15cm}
\begin{center}
\begin{tikzpicture}
\node at (-2,1.3) {$a^\prime_i$};
\node at (-2,1){$\bullet$};
\draw [dashed] (-1.8,1)--(-1.2,1);
\node at (-1,1.3) {$c^\prime_{ki}$};
\node at (-1,1) {$\bullet$};
\draw  [dashed](-0.8,1)--(-0.2,1);
\node at (0,1) {$\circ$};
\draw [dashed](0.2,1)--(0.8,1);
\node at (1,1.3) {$c_k$};
\node at (1,1){$\bullet$};
\draw [dashed](1.2,1)--(1.8,1);
\node at (2,1){$\circ$};
\node at (3,1){$\bullet$};
\node at (3,1.3) {$c^{\prime\prime}_{ki}$};
\draw [dashed] (2.2,1)--(2.8,1);
\node at (4,1){$\bullet$};
\node at (4,1.3){$b^\prime_{ji}$};
\draw [dashed](3.2,1)--(3.8,1);

\node at (-2,-0.2) {$a_i$};
\node at (-2,-0.5){$\bullet$};
\draw [dashed] (-1.9,-0.45)--(-0.1,0.45);
\node at (-1,0) {$\circ$};
\node at (0,0.5) {$\circ$};
\draw [dashed](0.1,0.55)--(0.9,0.95);
\end{tikzpicture}
\end{center}
\vspace{-0.15cm}
\item Next, we consider each of the remaining $a_i\in A$ and we select again some $b_{ji}$ with which it is connected. We take the path in $\mathcal{G}$ connecting them, and if it does not cross any path in $\mathcal{G}$ that generated an arc in $\mathcal{H}$, then we proceed as in the first case. Otherwise, we consider the first arc $\langle c^\prime_{ki}, c^{\prime\prime}_{ki}\rangle$ in $\mathcal{H}$ ``traversed'' by the new path. If $c_{ki}$ is the common node to the two involved paths, then we add it to $\mathcal{H}$ and remove the arc $\langle c^\prime_{ki}, c^{\prime\prime}_{ki}\rangle$ adding instead two new arcs $\langle c^\prime_{ki}, c_{ki}\rangle$ and $\langle c_{ki}, c^{\prime\prime}_{ki}\rangle$, together with the arc $\langle a_i, c_{ki}\rangle$.
\end{compactitem}
Finally, we proceed in the same way, but going backwards in the graph, for every $b_j\in B$ that was not still reached from any $a_i\in A$. Clearly, at any stage of the construction we add two new arcs in the worst case, and besides we need an idempotency step each time we consider a path that crosses $\mathcal{H}$. This finally produces a sequence from $t$ to $t^\prime$ with at most $3(m+n-2)+1$ steps, whose cost is not bigger than that of the original sequence.
Since the width of the trees change at most one at any step, the results about the bound of these widths follows immediately.
\end{proof}

The important thing about the bounds obtained above, is that they only depend on the cardinality of the multisets $A$ and $B$, but no at all on the properties of the domain of actions $(\mathds{A}, \bf{d})$. Even more, we can extend this result to any two finite trees $t=\sum_{a\in A} at_a$ and $t^\prime=\sum_{b\in B} bt_b$: the size and complexity of the ``continuations'' $t_a$ and $t_b$ do not compromise the bound on the number of first level steps of a sequence relating $t$ and $t^\prime$, that bound only depends on $||t||_1$ and $||t^\prime||_1$. 

\begin{proposition}\label{cota_total}
Let $t=\sum_{i=1}^n a_it_i$ and $t^\prime=\sum_{j=1}^m b_jt_j$ be two trees such that $t\equiv_{d}^{\alpha} t^\prime$. Then, we can prove this by means of a coinductive transformation sequence $\mathcal{C}$, that has at most $3(m+n-2)+1$ first level steps.
\end{proposition}
\begin{proof}
We observe that the result in Prop.~\ref{cota_primerpiso} could be obtained in exactly the same way if instead of a distance function on $A$, we would consider a relation $d\subseteq A\times A\times \mathds{R}^+$ that satisfies the properties that define ``bounds for a distance'' in $A$:
\begin{compactitem}
\item $\forall a\in A$ $d(a,a,d)$ $\forall d\in \mathds{R}^+$.
\item $d(a,b,d)$ $\Leftrightarrow$ $d(b,a,d)$.
\item $d(a,b,d_1) \wedge d(b,c,d_2)$ $\Rightarrow$ $d(a,c,d_1+d_2)$.
\end{compactitem}
Then, we consider the set of ``prefixed'' trees $\mathds{A}FTrees=\mathds{A}\times FTrees$, and we take $d_{\alpha}(at,bt^\prime,d_1+d_2)$ $\Leftrightarrow$ $d_1={\bf d}(a,b)\wedge t\rightsquigarrow_{\alpha,d_2/\alpha}t^\prime$. It is clear that for any distance $d$ on $A$, and any $\alpha \in \mathds{R}^`$, each such $d_\alpha$ defines bounds for a distance in $A$. Now, we can see each finite trees $t=\sum_{i=1}^n a_it_i$ as a $1$-depth trees $t=\sum_{i=1}^n\langle a_i,t_i\rangle$, for the alphabet $\mathds{A}FTrees$. For any $\alpha\in \mathds{R}^+$ we can consider the ``bound for a distance'' relation $d_\alpha$, and apply Prop.~\ref{cota_primerpiso} to conclude the proof.
\end{proof}

\subsection{An alternative proof of the bounds for the first level}
Even if the bounds obtained in Prop.~\ref{cota_primerpiso} are rather satisfactory we have been able to prove some tighter bounds by reducing the induced graph by means of local simplifications rules. We consider interesting to show this proof because they bring to light how we could proceed when transforming the trees all along the sequences (even if we have not been able to successfully use this ideas in order to fully prove the desired continuity result).


\begin{proposition} \label{prop:from_tsc_to_tbwc}
Any multi-stage graph that is tsc can be turned into a tbwc multi-stage graph
that is a subgraph of the original one.
\end{proposition}
\begin{proof}
We repeatedly remove any intermediate node which is not two ways connected, and it is clear that after any such removal the graph remains tsc.
\end{proof}

\begin{proposition} \label{pr}
\begin{compactenum}
\item The multi-stage graph associated to a sequence proving that $A\equiv^{\alpha}_dB$ is tbwc.
\item \label{segundo} Any subgraph of such a multi-stage graph, obtained by the removal of some internal nodes, that still remains tbwc, can also be obtained from some sequence proving $A\equiv^{\alpha}_d B$,
which will be shorter or have the same length than the initial sequence.
\end{compactenum}
\end{proposition}
\begin{proof}
(1) Trivial. (2) When removing an internal node we are possibly removing an application of idempotency that turns into a simple identity arc. In such a case, we can remove this stage of the sequence getting a shorter sequence still proving  $A\equiv^{\alpha}_d B$.
\end{proof}

\begin{remark}
The result in Prop.~\ref{pr}.\ref{segundo} is true, not only if we remove nodes that are not both ways connected (in fact, there is no such in a tbwc multi-stage graph!), but also if we remove any subset of intermediate nodes, as far as the graph remains tsc. This is what we next use in order to reduce the size of the graph.
\end{remark}

\begin{definition}
We say that a tbwc multi-stage graph is a \emph{diabolo} if there exists some stage  $i\in \{0\ldots m\}$ with $|S_i|=1$, such that:
\begin{inparaenum}[$(i)$]
\item for all $j<i$ and all $s_j\in S_j$, $|\{s^\prime \mid (s_j,s^{\prime})\in T\}|=1$, and
\item for all $j>i$ and all $s_j \in S_j$, $|\{s^\prime \mid (s^{\prime},s_j)\in T\}|=1$. We say that $s_i \in S_i$ is a center of the diabolo.
\end{inparaenum}
\end{definition}

\begin{figure} [t]
\begin{center}
\scriptsize{
\begin{tikzpicture}[->,>=stealth',shorten >=1pt,auto]
\matrix [matrix of math nodes, column sep={1cm,between origins},row sep={0.5cm,between origins}]
{
\node(1){\bullet};&\node () {}; & \node (4) {\bullet};&\node(){};&\node(7){\bullet};&\node(9){\bullet};&\node(){};&\node(13){\bullet};&\node(){};&\node(){};&\node(16){\bullet};\\
\node(){};&\node (2) {\bullet}; &\node(){};&\node(6){\bullet};&\node(){};&\node(){};&\node(11){\bullet};&\node(){};&\node(14){\bullet};&\node(15){\bullet};&\node(){};\\
\node(3){\bullet};&\node(){};&\node(5){\bullet}; &\node(){};&\node(8){\bullet};&\node(10){\bullet};&\node(){};&\node(12){\bullet};&\node(){};&\node(){};&\node(17){\bullet};\\
};

\begin{scope}[every node/.style={font=\scriptsize\itshape}]
\path [-]   
            (1) edge (2)
            (2) edge (3)
            (2) edge (4)
            (2) edge (5)
            (6) edge (7)
            (6) edge (8)
            (9) edge (11)
            (10) edge (11)
            (12) edge (14)
            (13) edge (14)
             (14) edge (15)
            (15) edge (16)
            (15) edge (17);
\end{scope}
\end{tikzpicture}}
\end{center}
\caption{Schematic examples of diabolos}
\end{figure}
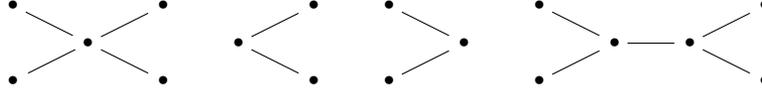

Next we prove that we can reduce any tbwc multi-stage graph into a disjoint union of diabolos. We will do it by removing some ``redundant'' arcs, in such a way that these removals do not destroy the \emph{tsc}
character of the graph. 

\begin{definition}
\begin{compactenum}
\item We say that a sequence of consecutive arcs \vspace{-0.2cm}$$(s_j,s_{j+1}), (s_{j+1},s_{j+2})\ldots (s_{j+k-1},s_{j+k})\vspace{-0.2cm}$$ in a tsc graph is a \emph{join sequence} if:
\begin{inparaenum}[$(i)$]
\item for all intermediate nodes in the sequence $s_{j+l}$ with $l \in \{1, \ldots, k-1\}$, the two arcs in the sequence, $(s_{j+l-1}, s_{j+l})$ and $(s_{j+l},s_{j+l+1})$ are the only ones in $T$ that involve $s_{j+l}$; and
\item there are at least two other arcs $(s_j, s^{\prime})$ and $(s^{\prime\prime},s_{j+l})$ in $T$ that are not in the sequence.
\end{inparaenum}
\item We say that a node $s_j\in S_j$ is left-reducible (resp.\ right-reducible) if it can only be reached from a single node $s_0\in S_0$ (resp. $s_m\in S_m$), but there are several arcs $(s^\prime_{j-1},s_j)\in T$ (resp. $(s_j,s^\prime_{j+1})\in T$).
\end{compactenum}
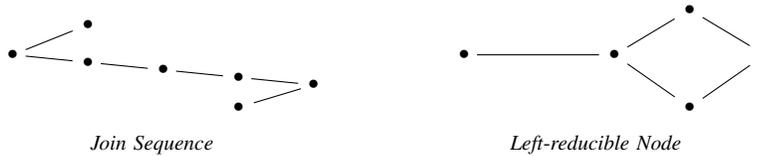
\begin{center}
\scriptsize{
\begin{tikzpicture}[->,>=stealth',shorten >=1pt,auto]
\matrix [matrix of math nodes, column sep={1cm,between origins},row sep={0.1cm,between origins}]
{
\node(){};&\node(){};&\node(){};&\node(){};&\node(){};&\node(){};&\node(){};\node(){};&\node () {}; & \node(){};&\node(4){\bullet};&\node(){};\\
\node(){};&\node () {}; & \node(){};&\node(){};&\node(){};&\node(){};&\node(){};&\node(){};&\node(){};&\node(){};&\node(){};&\node(){};&\node(){};\\
\node(){};&\node(8){\bullet};&\node(){};&\node(){};&\node(){};&\node(){};&\node(){};\node(){};&\node () {}; & \node(){};&\node(){};&\node(){};\\
\node(){};&\node(){};&\node(){};&\node(){};&\node(){};&\node(){};&\node(){};\node(){};&\node () {}; & \node(){};&\node(){};&\node(){};\\
\node(){};&\node () {}; & \node(){};&\node(){};&\node(){};&\node(){};&\node(){};&\node(){};&\node(){};&\node(){};&\node(){};&\node(){};&\node(){};\\
\node () {}; & \node(){};&\node(){};&\node(){};&\node(){};&\node(){};&\node(){};&\node(){};&\node(){};&\node(){};&\node(){};&\node(){};\\
\node(6){\bullet};&\node(){};&\node(){};&\node(){};&\node(){};&\node(){};&\node(){};\node(1){\bullet};&\node () {}; & \node(3){\bullet};&\node(){};&\node(5){\bullet};\\
\node(){};&\node(7){\bullet};&\node(){};&\node(){};&\node(){};&\node(){};&\node(){};\node(){};&\node () {}; & \node(){};&\node(){};&\node(){};\\
\node(){};&\node(){};&\node(8bajo){\bullet};&\node(){};&\node(){};&\node(){};&\node(){};\node(){};&\node () {}; & \node(){};&\node(){};&\node(){};\\
\node(){};&\node(){};&\node(){};&\node(9){\bullet};&\node(){};&\node(){};&\node(){};\node(){};&\node () {}; & \node(){};&\node(){};&\node(){};\\
\node(){};&\node(){};&\node(){};&\node(){};&\node(10){\bullet};&\node(){};&\node(){};\node(){};&\node () {}; & \node(){};&\node(){};&\node(){};\\
\node(){};&\node(){};&\node(){};&\node(){};&\node(){};&\node(){};&\node(){};\node(){};&\node () {}; & \node(){};&\node(){};&\node(){};\\
\node(){};&\node () {}; & \node(){};&\node(){};&\node(){};&\node(){};&\node(){};&\node(){};&\node(){};&\node(){};&\node(){};&\node(){};&\node(){};\\
\node(){};&\node(){};&\node(){};&\node(8bis){\bullet};&\node(){};&\node(){};&\node(){};\node(){};&\node () {}; & \node(){};&\node(4bis){\bullet};&\node(){};\\
\node(){};&\node () {}; & \node(){};&\node(){};&\node(){};&\node(){};&\node(){};&\node(){};&\node(){};&\node(){};&\node(){};&\node(){};&\node(){};\\
\node(){};&\node () {}; & \node(){};&\node(){};&\node(){};&\node(){};&\node(){};&\node(){};&\node(){};&\node(){};&\node(){};&\node(){};&\node(){};\\
\node(){};&\node () {}; & \node(){};&\node(){};&\node(){};&\node(){};&\node(){};&\node(){};&\node(){};&\node(){};&\node(){};&\node(){};&\node(){};\\
\node(){};&\node () {}; & \node(){};&\node(){};&\node(){};&\node(){};&\node(){};&\node(){};&\node(){};&\node(){};&\node(){};&\node(){};&\node(){};\\
\node(){};&\node(){\hspace{1.7cm}\mathit{Join\,\;Sequence}};&\node(){};&\node(){};&\node(){};&\node(){};\node(){};&\node () {}; & \node(){};&\node(){\hspace{-.5cm}\mathit{Left} \emph{-} \mathit{reducible\,\;Node}};&\node(){};\\
};

\begin{scope}[every node/.style={font=\scriptsize\itshape}]
\path [-]   
            (1) edge (3)
            (3) edge (4)
            (3) edge (4bis)
            (4) edge (5)
            (4bis) edge (5)
            (6) edge (8)
            (6) edge (7)
            (7) edge (8bajo)
            (8bajo) edge (9)
            (9) edge (10)
            (10) edge (8bis);
\end{scope}
\end{tikzpicture}}
\end{center}
\end{definition}

\begin{proposition} \label{eliminating_diamonds_joinsequence}
\begin{compactenum}
\item Whenever we have a join sequence in a tbwc graph, we can remove all its arcs and the intermediate nodes, preserving the tbwc property.
\item Whenever we have a left (resp.\ right)-reducible node $s_j$ in a tbwc, we can remove one of the arcs $(s^\prime_{j-1},s_j)$ reaching (resp. $(s_j,s^\prime_{j+1})$ leaving) the node, preserving the tsc property.
\end{compactenum}
\end{proposition}
\begin{proof}
(1) Clearly, after the removal the graph remains tbwc, due to the existence of the two ``lateral'' arcs $(s_j, s^{\prime}_{j+1})$ and $(s^\prime_{j+l-1},s_{j+l})$.
(2) It is clear that after the removal we have still another arc reaching (resp. leaving) the reducible node from the same side. This can be still used to reach the corresponding node $s_0\in S_0$ (resp. $s_m\in S_m$). And no other node in either $S_0$ and $S_m$ is affected by the removal.
\end{proof}

\begin{theorem}\label{reduce_to_collection_diabolo}
By removing some intermediate nodes we can reduce any tbwc multi-stage graph into another such graph (with smaller or equal total cost) which is the disjoint union of a collection of diabolos.
\end{theorem}
\begin{proof}
We start by reducing the graph applying Prop.~\ref{eliminating_diamonds_joinsequence}, until we cannot apply it anymore. We obtain a tsc graph
connecting the same sets of nodes $S_0$ and $S_m$, which can be turned into another tbwc (smaller) by applying Prop.\ref{prop:from_tsc_to_tbwc}. By abuse of notation, let us still denote by $S_1,\ldots, S_{m-1}$ the other stages of the graph.
\begin{compactitem}
\item First, we consider the connected components of the graph containing
  exactly one node $s_0$ belonging to $S_0$.
  These components are right-degenerated diabolos with $s_0$ as center.
  Indeed, if this was not the case, then the component would still contain a
  left-reducible node.
\item Now, let us consider a connected component containing a subset of nodes $\{s^1_0,\ldots, s^k_0\}\subseteq S_0$, with $k>1$.
It is not difficult to verify that for each $s^i_0$ ($1 \leq i \leq k$),
there is exactly one arc leaving $s_0^i$: otherwise, the connected component
would still contain a join sequence (due to the presence of multiple nodes in
$S_0$).
Using a similar argument, it is possible to prove that there can be only a single arc leaving each node in the following stages as well, until a stage $i$ is reached such that the component contains exactly one node $s_i\in S_i$.
Reasoning in a symmetric way from the right side (set $S_m$), it is eventually
possible to show that the considered component is a diabolo.
%
\end{compactitem}
Therefore, after the reduction, the graph is indeed the union of a family of disjoint diabolos.
\end{proof}

\begin{proposition}\label{size_diabolo}
\begin{compactenum}
\item \label{primero} A diabolo connecting two sets of nodes $S_0$ and $S_m$ satisfies $|S_i|\leq max\{|S_0|,|S_m|\}$ for each $i$.
\item Any disjoint union of diabolos connecting $S_0$ and $S_m$ satisfies $|S_i|\leq |S_0|+|S_m|$  for each $i$.
\item If the multi-stage graph corresponding to a sequence proving $A\equiv^\alpha_d B$ is a diabolo, then we can obtain another such sequence whose length will be at most $3(|A|+|B|-2)+1$.
\end{compactenum}
\end{proposition}
\begin{proof}
(1) Obvious. (2) We could think that this is an immediate consequence of \ref{primero}, but this is not always the case. Let us consider the disjoint union $S$ of two three-stages (degenerated) diabolos $S^1$ and $S^2$, with $|S^1_0|=1=|S^2_2|$; $|S^2_0|=8=|S^1_2|$; $|S^1_1|=5=|S^2_1|$. Then we have $S_0=S^1_0\cup S^2_0$, $S_1=S^1_1\cup S^2_1$, $S_2=S^1_2\cup S^2_2$ and therefore $|S_1|=10$, but $|S_0|=|S_2|=9$. As a consequence, the result would be wrong if we put $max$ instead of $+$. However, what we asserted is true in general: Let $S$ the disjoint union of a family of diabolos $S^j$ with $j\in\{1..k\}$, then we have $S_i=\bigcup_{j=1}^k S^j_i$, and as a consequence of \ref{primero}, $|S^j_i|\leq |S^j_0|+|S^j_m|$, and therefore $|S_i|=\sum_{j=1}^k |S^j_i|\leq \sum_{j=1}^k |S^j_0|+\sum_{j=1}^k |S^j_m|= |S_0|+|S_m|$.
(3) Whenever we have two relabeling steps at the same place, with no idempotency steps between them affecting that place, we can join them into a single relabeling step, without increasing the cost of the full transformation, because of triangular transitivity. As a consequence, we would have at most two relabeling steps for each idempotency step, and we have exactly $|A|+|B|-2$ idempotency steps at each diabolo. So we have, at the moment, at most $3(|A|+|B|-2)$ steps; possibly, we will need only one more relabeling step at the center of the diabolo. This gives us the bound $3(|A|+|B|-2)+1$.
\end{proof}

\begin{corollary} \label{size}
If we can prove $A\equiv^\alpha_d B$, then we can do it by means of a sequence whose intermediate sets satisfy $|S_i|\leq |A|+|B|$, and has at most $3(|A|+|B|-2)+1$ steps.
\end{corollary}
\begin{proof}
We apply Th.~\ref{reduce_to_collection_diabolo} to reduce the multi-stage graph corresponding to the sequence proving $A\equiv^\alpha_d B$. Then apply Prop.~\ref{size_diabolo} to get the bounds for the values $|S_i|$, and that for the number of steps.
\end{proof}

\subsection{Now we go down into the second level}
Once we have the base case of an inductive proof, we would like to proceed with the inductive case. When we thought that we had it, we decided to present first the particular case of the trees with only two levels, because its (bigger) simplicity would help the readers to understand the quite involved proof for the general case. Next we present the proof for this particular case.
Starting from $p=\sum_{i\in I}a_ip_i$ and $q=\sum_{j\in J}b_jq_j$, once we have proved our Cor.~\ref{size}, we can assume that the first level steps in the sequences proving $\pi_n(p)\equiv^\alpha_d \pi_n(q)$ are always the same and satisfy the bounds in the statement of Th.~\ref{reduce_to_collection_diabolo}. In order to study in detail the second level steps in these sequences, we need to see how the summands $p_i$ evolve along those sequences. 

Once we had some bound for the width of any process $p^\prime$ obtained by the evolution of the summands $p_i$, and another one for the length of the subsequences producing their evolution, adding all these bounds and that corresponding to the first level, we would obtain the bound for the two first levels together. 
The following Prop.~\ref{bounding_sequence} proves a preliminary result. It says that whenever we have a sequence proving $p\equiv^\alpha_d q$ with a ``limited number'' of first level steps, but $q$ contains summands $aq_i$ where $||q_i||$  is ``very large'', then we can prune these summands, getting some  $q^\prime$ for which $p\equiv^\alpha_d q^\prime$ can be proved by means of a sequence that only contains processes $p^i$ for which $||p^i||_2$ is ``moderately large'', and for any process $q^{\prime\prime}$ ``between'' $q^\prime$ and $q$ (that means that $q^{\prime\prime}$ can be obtained by adding some branches to some subprocesses of $q^\prime$, and also $q$ can be obtained from $q^{\prime\prime}$ in this way) we also have $p\equiv^\alpha_d q^{\prime\prime}$. Again, in order to make easier the comprehension of the proposition, we first present a lemma that covers the particular case corresponding to intermediate stages of a sequence connecting two processes $p$ and $q$ with $||p||$ and $||q||$ ``small''.

\begin{lemma}\label{lem}
For any intermediate process $p^i$ in the sequence $\mathcal{C}$ proving $p=\sum_{i=1}^n a_ip_i\equiv^\alpha_d q=\sum_{j=1}^m b_jq_j$, we can decompose $p^i$ into $p^i=p^{i1}+p^{i2}$ in such a way that $||p^{i1}||\leq ||p||+||q||$, and we have a sequence $\mathcal{C}^\prime:=p\equiv^\alpha_{d_1}p^{i1}\equiv^\alpha_{d_2} q$ with $d=d_1+d_2$, which has at most $3(|I|+|J|-2)+1$ first level steps, and only uses intermediate processes $r=\sum_{k\in K} c_kr_k$ with $||r||\leq ||p||+||q||$.\\
Moreover, for any decomposition $p^{i2}=p^{i3}+p^{i4}$, we can obtain a sequence $\mathcal{C}^{\prime\prime}:=p\equiv^\alpha_{d_1}p^{i1}+p^{i3}\equiv^\alpha_{d_2} q$ which has at most $(3(|I|+|J|-2)+1)*(||p^{i3}||+1)$ first level steps, and only uses intermediate processes $r=\sum_{k\in K} c_kr_k$, with $||r||\leq ||p||+||q||+||p^{i3}||$.
\end{lemma}
\begin{proof}
The first part is in fact a new formulation of Th.~\ref{reduce_to_collection_diabolo}, observing that each node at the $i$-th stage of the multi-stage graph induced by $\mathcal{C}$ corresponds to a summand of the corresponding process $p^i$. When reducing the multi-stage graph, we are pruning some of these summands. Therefore, the obtained intermediate processes at the sequence $\mathcal{C}^\prime$ are the processes $p^{i1}$ of the searched decomposition of $p^i$. Then, the remaining summands in $p^{i2}$ correspond to the nodes from the $i$-th stage of the multi-stage graph that were removed when reducing it.
It is easy to see that any of these summands can be ``reset'' into the multi-stage graph by means of a path that will connect the corresponding node with the two sides of the original multi-stage graph. This requires $(3(|I|+|J|-2)+1)$ additional arcs, and therefore $(3(|I|+|J|-2)+1)$ more steps in the sequence $\mathcal{C}^{\prime\prime}$, while the width of the intermediate processes in it is increased at most by $1$, when adding each of those paths.
\end{proof}

\begin{proposition}\label{bounding_sequence}
For each $f\in\mathds{N}$, there is a constant $C_{2,f}$ such that if we have a sequence $\mathcal{C}$ proving $p\equiv^\alpha_d r$ with $f$ first level steps and $r=\sum_{k\in K} c_kr_k$, with $r_k=\sum_{l\in K_k} c_{k,l}r_{k,l}$, then we can obtain some $r^\prime=\sum_{k\in K} c_k r^\prime_k$ with $r^\prime_k=\sum_{l\in L_k} c_{k,l} r^\prime_{k,l}$, where $L_k\subseteq K_k$, with $|L_k|\leq 2^f ||p||$, whose intermediate processes $p^i$ are of the form $p^i=\sum a^i_jp^i_j$, with $p^i_j=\sum _{k\in K^i_j}b_kp^i_{j,k}$ and $|K^i_j|\leq 2^f ||p||$, and $length(S^\prime)\leq C_{2,f}$.\\
Moreover, taking $m\in\mathds{N}$, we also have a family of constants $C_{2,f,m}$ such that for any $r^{\prime\prime}=\sum c_kr^{\prime\prime}_k$ with $r^{\prime\prime}_k=\sum_{l\in L^\prime_k} c_{k,l}r_{k,l}$, where $L_k\subseteq L^\prime_k\subseteq K_k$, we can also prove $p\equiv^\alpha_d r^{\prime\prime}$, by means of a sequence $\mathcal{C}^{\prime\prime}$ whose intermediate processes $p^i$ satisfy $||p^i||_2\leq max\{2^f ||p||, |L^\prime_k|_{k\in K}\}$ and $length(S^{\prime\prime})\leq C_{2,f,max\{|L^\prime_k|_{k\in K}\}}$. 
\end{proposition}

\begin{theorem}\label{step_bound_level2}
For all $k\in\{1,2\}$ and $w\in\mathds{N}$ there exists a bound $lb(k,w)\in\mathds{N}$ such that for all $p,q$ with $||p||_k$, $||q||_k\leq w$ and $p\equiv^\alpha_d q$ we can prove the latter by means of a sequence $\mathcal{C}^\prime$ that has no more than $lb(k,w)$ steps in each of its two first levels.
\end{theorem}
\begin{proof}
At the same time that the existence of the bound $lb(k,w)$ we will prove a bound $wb(k,w)$ for the width $||p^i||_2$ of any process $p^i$ along the $2$-unfolding of the sequence $\mathcal{C}^\prime$.\\
\underline{k=1} Th.~\ref{reduce_to_collection_diabolo} just states the result for this case, taking $wb(1,w)=2w$ and $lb(1,w)\leq 6w$. \\
\underline{k=2} Let $\mathcal{C}$ be a sequence proving $p\equiv^\alpha_d q$ which has less than $lb(1,w)$ steps at the first level. If $\mathcal{C}$ contains no first level step, then we have $p=\sum_{i=1}^n a_ip_i$, $q=\sum_{i=1}^n a_i q_i$ and $\mathcal{C}$ can be factorized into a collection of sequences $(\mathcal{C}_i^\prime)_{i=1}^n$ proving $p_i \equiv_{d_i} q_i$. 
Then, we only need to apply Th.~\ref{reduce_to_collection_diabolo} to each of these sequences, so that we could take $wb(2,w)=2w$ and $lb(2,w)=2w*6w=12w^2$.\\
If $\mathcal{C}$ contains some first level step, we select any of them and divide the $2$-unfolding of $\mathcal{C}$ into $\mathcal{C}^1\circ \langle s \rangle\circ\mathcal{C}^2$, where by abuse of notation we identify the sequences and their $1$-unfolding. Let us consider the case in which the central step corresponds to a relabeling $ap_i^{\prime}\rightarrow bp_i^\prime$ (the other cases are analogous). Assume that $\mathcal{C}^1$ proves $p\equiv^\alpha_{d_1}p^\prime$ and $\mathcal{C}^2$ proves $q^\prime\equiv^\alpha_{d_3}q$, with $d=d_1+{\bf d}(a,b)+d_3$. We can apply Prop.~\ref{bounding_sequence} to both sequences, taking $\mathcal{C}^2$ in the opposite direction, getting $r^\prime$ and $r^{\prime\prime}$, with $\mathcal{C}^{\prime1}$ proving $p\equiv^\alpha_{d_1} r^\prime$ and $\mathcal{C}^{\prime2}$ proving $r^{\prime\prime}\equiv^\alpha_{d_3} q$, that have the same steps at the first level than $\mathcal{C}^1$ and $\mathcal{C}^2$, respectively, and only use intermediate processes $p^i$ with $||p^i||_2\leq2^{lb(1,w)-1}*w$. The result is also valid for any intermediate process between $r_1$ and $p^\prime$, and anyone between $r_2$ and $q^\prime$, but increasing the bounds in the adequate way. In particular, for the process $r_1+r_2$, we have $p\equiv^\alpha_{d_1} r_1+r_2$ and $r_1+r_2\equiv^\alpha_{d_3} q$, by means of two sequences that only include intermediate processes $p^i$ with $||p^i||_2\leq 2^{lb(1,w)}*w$. Joining these two sequences with the step $s$ we obtain the sequence $\mathcal{C}^\prime$. Adding the bounds for the corresponding lengths of the $2$-unfolding of the two sequences obtained by application of  Prop.~\ref{bounding_sequence}, we obtain the bound $lb(2,w)$. To be more precise, the definitive value of $lb(2,w)$ will be the maximum of the bounds for the different cases (there are finitely many) considered above.
\end{proof}

Certainly, the proof is quite involved and difficult to follow, but it works in this case. Unfortunately, when we tried to develop it for the general case we discovered that some details failed, or at least cannot be justified in the same way. Let us roughly explain why: the argument used in Lemma~\ref{lem} looks for an small ``kernel'' $p^{i1}$ of any intermediate process along a sequence. Then Prop.~\ref{bounding_sequence} joins the kernels obtained when considering the two subsequences starting at any side of the original sequence and ending at any intermediate process. The second part of Lemma~\ref{lem} says us that we can get in this way a sequence satisfying the desired bounds. But when we are in a deeper level, a duplication (by using idempotency) at a lower level could cause an undesired increasing of the ``closure'' of the union of the two (one from each side) kernels. And this problem could appear over and over (at least, we are not able to prove that this is not the case). Therefore, we have to leave this problem open, at the moment, any help?

\section{Conclusions and future work}
The first two authors of this paper are involved in a detailed study of bisimulation distances, that will constitute in fact the forthcoming Ph.D. Thesis of the first one, supervised by the second. As we have already pointed out, while the global bisimulation distance we propose is more difficult to manage than other existing bisimulation distances, we claim that it provides a much more accurate measure of the differences between two processes. Moreover, our approach has many nice properties, and we still think that continuity with respect to projections is one of them. Even if we have not been able yet to present a full continuity proof, we think that, after having invested a great number of hours in this quest, time has come to present here the (partial) results we got so far.

We are now working on several extensions of our approach. In particular the modal interface framework \cite{lnw07, lv13, rbbclp11} is a quite suggestive one, where we are obtaining quite promising results. We are also interested in stochastic distances \cite{fpp11} and those based on logics \cite{dgjp04}. When considering probabilistic choices between branches associated with the same action, these two notions of distance \emph{do} capture global differences. However, this is not the case when choices between different actions are considered. It is our intention to provide a ``fully global'' probabilistic distance (covering also the latter case) that will extend our global distance to these frameworks. Finally, we recently started to consider interval temporal logics~\cite{beatcs11},
aiming at devising a suitable notion of global distance between interval models.

\bibliographystyle{eptcs}
\bibliography{bibliografia}
\end{document}